\documentclass[copyright]{eptcs}

\providecommand{\event}{FICS 2013}  % Name of the event you are submitting to

\usepackage{amsmath}
\usepackage{amssymb}
\usepackage{amsthm}
\usepackage[greek,english]{babel}
\usepackage{calc}
\usepackage[margin=10pt,font=small,labelfont=bf,labelsep=period,justification=raggedright]{caption}
\usepackage{comment}
\usepackage{dsfont}
\usepackage{float}
\usepackage[T1]{fontenc}
\usepackage{mathrsfs}
\usepackage{paralist}
\usepackage{pbox}
\usepackage{stmaryrd}
\usepackage{tikz}
\usetikzlibrary{arrows,matrix}

\theoremstyle{definition}

\newtheorem{definition}{Definition}[section]

\theoremstyle{plain}

\newtheorem{theorem}[definition]{Theorem}
\newtheorem{lemma}[definition]{Lemma}
\newtheorem{proposition}[definition]{Proposition}
\newtheorem{corollary}[definition]{Corollary}

\theoremstyle{remark}

\newtheorem{example}[definition]{Example}

\newtheoremstyle{customtheorem}
    {\topsep}
    {\topsep}
    {\itshape}
    {}
    {\bfseries}
    {.}
    { }
    {\thmnote{#3}}

\theoremstyle{customtheorem}

\definecolor{shadecolor}{rgb}{0.85, 0.85, 0.85}

\newcommand{\displaypunctuation}[1]{{\makebox[0pt]{\;#1}}}

\newcommand{\Iff}{\mathrel{\iff}}

\newcommand{\set}[1]{{\{{#1}\}}}
\newcommand{\Set}[2]{{\{{#1}\mathrel{\vert}{\textnormal{#2}}\}}}

\newcommand{\pair}[2]{{\langle{#1},{#2}\rangle}}

\newcommand{\sequence}[1]{{\langle{#1}\rangle}}
\newcommand{\Sequence}[2]{{\langle{#1}\mathrel{\vert}{\textnormal{#2}}\rangle}}

\newcommand{\union}{\cup}
\newcommand{\Intersection}{\mathop{\textstyle\bigcap}}

\DeclareMathOperator{\domain}{\mathsf{dom}}

\newcommand{\naturals}{\mathbb{N}}

\newcommand{\reals}{\mathbb{R}}

\renewcommand{\to}{\rightarrow}

\newcommand{\meet}{\sqcap}

\DeclareMathOperator{\oh}{\mathsf{oh}}

\newbox\transitionbox

\newcommand{\transition}[1][]{                          % \usepackage{tikz}
    \setbox\transitionbox=\hbox{\ensuremath{\scriptstyle #1}}
    \pgfmathsetlengthmacro{\transitionlength}{max(0.925em, \wd\transitionbox) + .65em}
    \mathrel{
        \begin{tikzpicture}[baseline = -.57ex]

            \draw[line cap = round, ->, shorten < = .08em, shorten > = .045em] 
                (0,0) --    node[inner sep = 0pt, minimum height = 0pt, minimum width = 0pt, above = .49ex]   {\box\transitionbox}    (\transitionlength, 0);

        \end{tikzpicture}
    }
}

\newcommand{\notransition}[1][]{                        % \usepackage{tikz}
    \setbox\transitionbox=\hbox{\ensuremath{\scriptstyle #1}}
    \pgfmathsetlengthmacro{\transitionlength}{max(.925em, \wd\transitionbox) + .65em}
    \mathrel{
        \begin{tikzpicture}[baseline = -.57ex]
        
            \draw[line cap = round, ->, shorten < = .08em, shorten > = .045em] 
                (0,0) --    node[inner sep = 0pt, minimum height = 0pt, minimum width = 0pt, above = .49ex]   {\box\transitionbox}    (\transitionlength, 0);
                
            \draw[line cap = round]
                (\transitionlength / 2 - .11em, -.54ex) -- (\transitionlength / 2 + .11em, .54ex);
                
        \end{tikzpicture}
    }
}

\newcommand{\tagset}{{\mathrm{T}}}

\DeclareMathOperator{\fix}{\mathsf{fix}}

\DeclareMathOperator{\1m2}{\mathsf{1m2}}

\newcommand{\sortA}{\mathrm{A}}
\newcommand{\sortD}{\mathrm{D}}

\newcommand{\meetsymbol}{\boldsymbol{\meet}}
\newcommand{\sqsubseteqsymbol}{\boldsymbol{\sqsubseteq}}
\newcommand{\sqsubsetsymbol}{\boldsymbol{\sqsubset}}

\newcommand{\LUBsymbol}{\boldsymbol{\mathop{\textstyle\bigsqcup}\nolimits}}
\newcommand{\leqsymbol}{\boldsymbol{\leq}}
\newcommand{\lsymbol}{\boldsymbol{<}}

\newcommand{\zerosymbol}{\boldsymbol{0}}
\newcommand{\dsymbol}{\boldsymbol{\mathrm{d}}}

\newcommand{\structureA}{\mathfrak{A}}
\newcommand{\structureS}{\mathfrak{S}}

\newcommand{\carrier}[1]{{\mathopen{|}#1\mathclose{|}}}

\newcommand{\distance}[3]{{\dsymbol^{#1}(#2, #3)}}

\newcommand{\OmTsymbol}{\boldsymbol{\1m2}}

\title{
    The Fixed-Point Theory of Strictly Contracting Functions on Generalized Ultrametric Semilattices\thanks{
        This work was supported in part by the Center for Hybrid and Embedded Software Systems (CHESS) at UC Berkeley, 
        which receives support from the National Science Foundation (NSF awards \#0720882 (CSR-EHS: PRET), \#0931843 (CPS: Large: ActionWebs), and \#1035672 (CPS: Medium: Ptides)), 
        the Naval Research Laboratory (NRL \#N0013-12-1-G015), 
        and the following companies: 
        Bosch, National Instruments, and Toyota.
    }
}

\author{
    Eleftherios Matsikoudis
    \institute{University of California, Berkeley}
    \email{ematsi@eecs.berkeley.edu}
    \and
    Edward A. Lee
    \institute{University of California, Berkeley}
    \email{eal@eecs.berkeley.edu}
}

\def\titlerunning{The Fixed-Point Theory of Strictly Contracting Functions on Generalized Ultrametric Semilattices}

\def\authorrunning{E. Matsikoudis \& E. A. Lee}

\begin{document}

%%%%%%%%%%%%%%%%%%%%%%%%%%%%%%%%%%%%%%%%%%%%%%%%%%%%%%%%%%%%%%%%%%%%%%%%%%%%%%%%

\maketitle

%%%%%%%%%%%%%%%%%%%%%%%%%%%%%%%%%%%%%%%%%%%%%%%%%%%%%%%%%%%%%%%%%%%%%%%%%%%%%%%%

\begin{abstract}
    We introduce a new class of abstract structures,
    which we call generalized ultrametric semilattices,
    and in which the meet operation of the semilattice coexists with a generalized distance function in a tightly coordinated way.
    We prove a constructive fixed-point theorem for strictly contracting functions on directed-complete generalized ultrametric semilattices,
    and introduce a corresponding induction principle.
    We cite examples of application in the semantics of logic programming and timed computation,
    where, 
    until now, 
    the only tool available has been the non-constructive fixed-point theorem of Priess-Crampe and Ribenboim for strictly contracting functions on spherically complete generalized ultrametric semilattices.
\end{abstract}

%%%%%%%%%%%%%%%%%%%%%%%%%%%%%%%%%%%%%%%%%%%%%%%%%%%%%%%%%%%%%%%%%%%%%%%%%%%%%%%%

\section{Introduction} \label{sec:introduction}

Fixed-point semantics in computer science has almost invariably been based on the fixed-point theory of order-preserving functions on ordered sets, 
or that of contraction mappings on metric spaces. 
More recently, 
however,
there have been instances of fixed-point problems involving strictly contracting functions on generalized ultrametric spaces, 
such as in the semantics of logic programming 
(e.g., see \cite{Hitzler2003GMAUDLP}, \cite{Priess-Crampe2000USALP}), 
or the study of timed systems 
(e.g., see \cite{Naundorf2000SCFHAUFP}, \cite{Liu2006MTCS}), 
that are not amenable to classical methods (see \cite[thm.\,A.2 and thm.\,A.4]{Matsikoudis2013TFTOSCF}). 
Until recently, 
the only tool available for dealing with such problems was a non-constructive fixed-point theorem of Priess-Crampe and Ribenboim (see \cite{Priess-Crampe1993FPCAGPS}). 
But in \cite{Matsikoudis2013TFTOSCF}, 
a constructive theorem was obtained, 
tailored to the general form in which these problems typically appear in computer science,  
also delivering an induction principle for proving properties of the constructed fixed-points. 
What is interesting is that the proof of that theorem involved,
not just the generalized ultrametric structure of the spaces of interest,
but also a natural, inherent ordering of these spaces,
and more importantly,
the interplay between the two,
which was distilled in two simple properties of the following form:
\begin{enumerate}

    \item
        if $d(x_1, x_2) \leq d(x_1, x_3)$,
        then $x_1 \meet x_3 \sqsubseteq x_1 \meet x_2 \enspace$;

    \item
        $d(x_1 \meet x_2, x_1 \meet x_3) \leq d(x_2, x_3)$.

\end{enumerate}
As it turns out, 
these two simple properties imply all formal properties of the relationship between the generalized distance function and the order relation in those spaces 
(see \cite{Matsikoudis2013AAOTTOGUSOLS}).

The purpose of this work is to formulate the fixed-point theory of \cite{Matsikoudis2013TFTOSCF} as an abstract theory that can be readily applied to different fields and problems,  
such as the question of meaning of logic programs or the study of feedback in timed systems.
To this end, 
we introduce a new class of abstract structures, 
which we call \emph{generalized ultrametric semilattices}, 
prove a constructive fixed-point theorem of strictly contracting functions on directed-complete generalized ultrametric semilattices,
and introduce a corresponding induction principle.

\section{Generalized Ultrametric Semilattices} \label{sec:generalized_ultrametric_semilattices}

We assume that the reader is familiar with the concept of many-sorted signature,
which is, 
of course, 
a straightforward generalization of that in the one-sorted case 
(e.g., see \cite[chap.\,1.1]{Hodges1993MT}).

We write $\Sigma$ for a two-sorted signature consisting of two sorts $\sortA$ and $\sortD$, 
and the following symbols:
\begin{enumerate}
    
    \item
        an infix function symbol $\meetsymbol$ of type $\sortA \boldsymbol{\times} \sortA \boldsymbol{\to} \sortA$;
        
    \item
        an infix relation symbol $\leqsymbol$ of type $\sortD \boldsymbol{\times} \sortD$;
         
    \item
        a constant symbol $\zerosymbol$ of type $\boldsymbol{1} \boldsymbol{\to}\sortD$;
        
    \item
        a function symbol $\dsymbol$ of type $\sortA \boldsymbol{\times} \sortA \boldsymbol{\to} \sortD$.
        
\end{enumerate} 

\begin{definition}
    A $\Sigma$-structure is a function $\structureA$ from the set of sorts and symbols of $\Sigma$ 
    such that $\structureA(\sortA)$ and $\structureA(\sortD)$ are non-empty sets, 
    and the following are true:
    \begin{enumerate}
        
        \item
            $\structureA(\meetsymbol)$ is a function from $\structureA(\sortA) \times \structureA(\sortA)$ to $\structureA(\sortA)$;
            
        \item
            $\structureA(\leqsymbol)$ is a subset of $\structureA(\sortD) \times \structureA(\sortD)$;
            
        \item
            $\structureA(\zerosymbol)$ is a member of $\structureA(\sortD)$;
            
        \item
            $\structureA(\dsymbol)$ is a function from $\structureA(\sortA) \times \structureA(\sortA)$ to $\structureA(\sortD)$.
            
    \end{enumerate} 
\end{definition}

Assume a $\Sigma$-structure $\structureA$.

We write $\carrier{\structureA}_\sortA$ for $\structureA(\sortA)$, 
$\carrier{\structureA}_\sortD$ for $\structureA(\sortD)$, 
$\meetsymbol^\structureA$ for $\structureA(\meetsymbol)$,
$\leqsymbol^\structureA$ for $\structureA(\leqsymbol)$,
$\zerosymbol^\structureA$ for $\structureA(\zerosymbol)$, and
$\dsymbol^\structureA$ for $\structureA(\dsymbol)$.

We call $\carrier{\structureA}_\sortA$ the \emph{carrier} of $\structureA$ of sort $\sortA$, 
or the \emph{abstract set} of $\structureA$,
and $\carrier{\structureA}_\sortD$ the \emph{carrier} of $\structureA$ of sort $\sortD$, 
or the \emph{distance set} of $\structureA$.

It is, 
of course, 
possible to define concepts of homomorphism, substructure, etc., for $\Sigma$-structures as instances of the standard concepts homomorphism, substructure, etc., for many-sorted structures,
which are,
of course,
straightforward generalizations of those for one-sorted structures
(e.g., see \cite[chap.\,1.2]{Hodges1993MT})
(see \cite{Matsikoudis2013AAOTTOGUSOLS}).

The $\Sigma$-structures that we are interested in are those in which the function assigned to $\meetsymbol$ behaves as the meet operation of a semilattice, 
the function assigned to $\dsymbol$ as the generalized distance function of a generalized ultrametric space, 
and the two satisfy a couple of simple properties.

\begin{definition} \label{def:generalized_ultrametric_semilattice}
    A \emph{generalized ultrametric semilattice} is a $\Sigma$-structure $\structureA$ 
    such that the following are true:
    \begin{enumerate}

        \item \label{enum:generalized_ultrametric_semilattice_1}
            $\pair{\carrier{\structureA}_\sortA}{\meetsymbol^\structureA}$ is a semilattice\footnote{
                \label{foot:semilattice}
                For every set $S$, 
                and every binary operation $\meet$ on $S$, 
                $\pair{S}{\meet}$ is a \emph{semilattice} 
                if and only if 
                for any $s_1, s_2, s_3 \in S$,
                the following are true:
                \begin{enumerate}
                    
                    \item
                        $(s_1 \meet s_2) \meet s_3 = s_1 \meet (s_2 \meet s_3)$;
                        
                    \item
                        $s_1 \meet s_2 = s_2 \meet s_1$;
                    
                    \item
                        $s_1 \meet s_1 = s_1$.
                    
                \end{enumerate}
            };

        \item \label{enum:generalized_ultrametric_semilattice_2}
            $\sequence{\carrier{\structureA}_\sortD, \leqsymbol^\structureA, \zerosymbol^\structureA}$ is a pointed\footnote{
                An ordered set\footnotemark{} is \emph{pointed} 
                if and only if 
                it has a least element.
                We write $\sequence{P, \leqslant, 0}$ for a pointed ordered set $\pair{P}{\leqslant}$ with least element $0$.
            } ordered set;

        \footnotetext{
            An \emph{ordered set} is an ordered pair $\pair{P}{\leqslant}$
            such that $P$ is a set, 
            and $\leqslant$ is a reflexive, transitive, and antisymmetric binary relation on $P$.
        }

        \item \label{enum:generalized_ultrametric_semilattice_3}
            $\sequence{\carrier{\structureA}_\sortA, \carrier{\structureA}_\sortD, \leqsymbol^\structureA, \zerosymbol^\structureA, \dsymbol^\structureA}$ is a generalized ultrametric space\footnote{
                A \emph{generalized ultrametric space} is a quintuple $\sequence{A, P, \leqslant, 0, d}$ 
                such that $A$ is a set, 
                $\sequence{P, \leqslant, 0}$ is a pointed ordered set, 
                $d$ is a function from $A \times A$ to $P$,
                and for any $a_1, a_2, a_3 \in A$ 
                and every $p \in P$,
                the following are true:
                \begin{enumerate}
            
                    \item \label{enum:generalized_ultrametric_1}
                        $d(a_1, a_2) = 0$ 
                        if and only if 
                        $a_1 = a_2$;
            
                    \item \label{enum:generalized_ultrametric_2}
                        $d(a_1, a_2) = d(a_2, a_1)$;
            
                    \item \label{enum:generalized_ultrametric_3}
                        if $d(a_1, a_2) \leqslant p$ and $d(a_2, a_3) \leqslant p$,
                        then $d(a_1, a_3) \leqslant p$.
                        
                \end{enumerate}
                We refer to clause \ref{enum:generalized_ultrametric_1} as the \emph{identity of indiscernibles},
                clause \ref{enum:generalized_ultrametric_2} as \emph{symmetry},
                and clause \ref{enum:generalized_ultrametric_3} as the \emph{generalized ultrametric inequality}.
            };
        
        \item \label{enum:generalized_ultrametric_semilattice_4}
            for every $a_1, a_2, a_3 \in \carrier{\structureA}_\sortA$, 
            the following are true:
            \begin{enumerate}
                
                \item \label{enum:generalized_ultrametric_semilattice_4_1}
                    if $\dsymbol^\structureA(a_1, a_2) \leqsymbol^\structureA \dsymbol^\structureA(a_1, a_3)$, 
                    then $(a_1 \meetsymbol^\structureA a_3) \meetsymbol^\structureA (a_1 \meetsymbol^\structureA a_2) = a_1 \meetsymbol^\structureA a_3$;
                
                \item \label{enum:generalized_ultrametric_semilattice_4_2}
                    $\dsymbol^\structureA(a_1 \meetsymbol^\structureA a_2, a_1 \meetsymbol^\structureA a_3) \leqsymbol^\structureA \dsymbol^\structureA(a_2, a_3)$.
                    
            \end{enumerate}

    \end{enumerate}
\end{definition}

Notice that, 
in Definition~\ref{def:generalized_ultrametric_semilattice}.\ref{enum:generalized_ultrametric_semilattice_1}, 
a semilattice is viewed as an algebraic structure. 
For the most part, 
it will be more convenient to view a semilattice as an ordered set.\footnote{
    An ordered set $\pair{P}{\leqslant}$ is a \emph{semilattice} 
    (also called a \emph{meet-semilattice} or a {lower semilattice})
    if and only if 
    for any $p_1, p_2 \in P$, 
    there is a greatest lower bound (also called a meet) of $p_1$ and $p_2$ in $\pair{P}{\leqslant}$.
}
The two views are closely connected, 
and one may seamlessly switch between them
(e.g., see \cite[lem.\,2.8]{Davey2002ITLAO}).
Formally, 
it is simpler to work with a meet operation than with an order relation
(see \cite{Matsikoudis2013AAOTTOGUSOLS}).
But informally, 
we will recover the order relation from the meet operation, 
and for every $a_1, a_2 \in \carrier{\structureA}_\sortA$, 
write $a_1 \sqsubseteqsymbol^\structureA a_2$ 
if and only if 
$a_1 \meetsymbol^\structureA a_2 = a_1$.
In particular, 
we may rewrite Definition~\ref{def:generalized_ultrametric_semilattice}.\ref{enum:generalized_ultrametric_semilattice_4} in the following form:
\begin{enumerate}
    \setcounter{enumi}{3}
    
    \item 
        for every $a_1, a_2, a_3 \in \carrier{\structureA}_\sortA$, 
        the following are true:
        \begin{enumerate}
            
            \item 
                if $\dsymbol^\structureA(a_1, a_2) \leqsymbol^\structureA \dsymbol^\structureA(a_1, a_3)$, 
                then $a_1 \meetsymbol^\structureA a_3 \sqsubseteqsymbol^\structureA a_1 \meetsymbol^\structureA a_2$;
            
            \item 
                $\dsymbol^\structureA(a_1 \meetsymbol^\structureA a_2, a_1 \meetsymbol^\structureA a_3) \leqsymbol^\structureA \dsymbol^\structureA(a_2, a_3)$.
                
        \end{enumerate}

\end{enumerate}
Of course, 
all this can be done formally,
but we shall not worry ourselves over the details.

For notational convenience, 
we will informally write $\sqsubsetsymbol^\structureA$ for the irreflexive part of $\sqsubseteqsymbol^\structureA$, 
and $\lsymbol^\structureA$ for the irreflexive part of $\leqsymbol^\structureA$.

Assume a generalized ultrametric semilattice $\structureA$.

We say that $\structureA$ is \emph{directed-complete} 
if and only if 
$\pair{\carrier{\structureA}}{\sqsubseteqsymbol^\structureA}$ is directed-complete\footnote{
    An ordered set $\pair{P}{\leqslant}$ is \emph{directed-complete} 
    if and only if 
    every subset of $P$ that is directed\footnotemark{} in $\pair{P}{\leqslant}$ has a least upper bound in $\pair{P}{\leqslant}$.
}\footnotetext{
    For every ordered set $\pair{P}{\leqslant}$, 
    and every $D \subseteq P$,
    $D$ is \emph{directed} in $\pair{P}{\leqslant}$
    if and only if 
    $D \neq \emptyset$, 
    and every finite subset of $D$ has an upper bound in $\pair{D}{\leqslant_D}$, 
    where $\leqslant_D$ is the restriction of $\leqslant$ to $D$.
}. 

If $\structureA$ is directed-complete, 
then for every $D \subseteq \carrier{\structureA}_\sortA$ that is directed in $\pair{\carrier{\structureA}_\sortA}{\sqsubseteqsymbol^\structureA}$, 
we write $\LUBsymbol^\structureA D$ for the least upper bound of $D$ in $\pair{\carrier{\structureA}_\sortA}{\sqsubseteqsymbol^\structureA}$.

We say that $\structureA$ is \emph{spherically complete} 
if and only if 
$\sequence{\carrier{\structureA}_\sortA, \carrier{\structureA}_\sortD, \leqsymbol^\structureA, \zerosymbol^\structureA, \dsymbol^\structureA}$ is spherically complete\footnote{
    A generalized ultrametric space $\sequence{A, P, \leqslant, 0, d}$ is \emph{spherically complete} 
    if and only if 
    for every non-empty chain $C$ of balls\footnotemark{} in $\sequence{A, P, \leqslant, 0, d}$, 
    $\Intersection C \neq \emptyset$.
}\footnotetext{
    For every generalized ultrametric space $\sequence{A, P, \leqslant, 0, d}$,
    and every $B \subseteq A$,
    $B$ is a \emph{ball} in $\sequence{A, P, \leqslant, 0, d}$
    if and only if 
    there is $a \in A$ and $p \in P$ 
    such that $B = \Set{a' \in A}{$d(a', a) \leqslant p$}$.
}.

The paradigmatic example of a generalized ultrametric semilattice is the standard generalized ultrametric semilattice $\structureS[\pair{T}{\leq_T}, V]$ of all linear signals from some totally ordered set $\pair{T}{\leq_T}$ to some non-empty set $V$ 
(see \cite{Matsikoudis2013AAOTTOGUSOLS}). 
Indeed, 
the definition of generalized ultrametric semilattices was motivated by the fact that every generalized ultrametric semilattice with a totally ordered distance set is isomorphic to a standard generalized ultrametric semilattice of linear signals 
(see \cite[thm.\,2]{Matsikoudis2013AAOTTOGUSOLS}).

An example of a non-standard generalized ultrametric semilattice of linear signals is the set of all finite and infinite sequences over some non-empty set of values, 
equipped with the standard prefix relation and the so-called ``Baire-distance function'' 
(e.g., see \cite{Bakker1998DMFPLAOBFPT}). 

\begin{example} \label{exam:finite_and_infinite_sequences}
    Let $V$ be a non-empty set. 
    
    Let $\structureA$ be a $\Sigma$-structure 
    such that $\carrier{\structureA}_\sortA$ is the set of all finite and infinite sequences over $V$, 
    $\carrier{\structureA}_\sortD = \reals_{\geq 0}$,\footnote{
        We write $\reals_{\geq 0}$ for the set of all non-negative real numbers. 
    } 
    and the following are true:
    \begin{enumerate}
        
        \item
            ${\meetsymbol^\structureA}$ is a binary operation on $\carrier{\structureA}_\sortA$ 
            such that for every $s_1, s_2 \in \carrier{\structureA}_\sortA$, 
            $s_1 \meetsymbol^\structureA s_2$ is the greatest common prefix of $s_1$ and $s_2$;
                        
        \item
            ${\leqsymbol^\structureA}$ is the standard order on $\reals_{\geq 0}$; 
        
        \item
            $\zerosymbol^\structureA = 0$;
            
        \item
            $\dsymbol^\structureA$ is a function from $\carrier{\structureA}_\sortA \times \carrier{\structureA}_\sortA$ to $\carrier{\structureA}_\sortD$ 
            such that for every $s_1, s_2 \in \carrier{\structureA}_\sortA$, 
            \[
                \dsymbol^\structureA(s_1, s_2) = 
                \begin{cases}
                    0                                                                       &\text{if $s_1 = s_2$;} \\
                    2^{- \min \Set{n}{$n \in \naturals$ and $s_1(n) \not\simeq s_2(n)$}}    &\text{otherwise.\footnotemark{}}
                \end{cases}
            \]
        
    \end{enumerate}
    
    \footnotetext{
        We write $\naturals$ for the set of all natural numbers, 
        and $\leq_\naturals$ for the standard order on $\naturals$.
    }
    
%     Clearly, 
%     for every $s_1, s_2, s_3, s_4 \in \carrier{\structureA}_\sortA$, 
%     \[
%         \dsymbol^\structureA(s_1, s_2) \leqsymbol^\structureA \dsymbol^\structureA(s_3, s_4) 
%     \] 
%     if and only if 
%     \[
%         \mathrm{d}_{\signalset[\pair{\naturals}{\leq_\naturals}, V]}(s_1, s_2) \supseteq \mathrm{d}_{\signalset[\pair{\naturals}{\leq_\naturals}, V]}(s_3, s_4) \displaypunctuation{,}
%     \] 
%     and thus, 
%     $\structureA$ is a generalized ultrametric semilattice.
    It is easy to verify that $\structureA$ is a directed-complete and spherically complete generalized ultrametric semilattice.
\end{example}

Notice that the generalized ultrametric space associated with the generalized ultrametric semilattice $\structureA$ of Example~\ref{exam:finite_and_infinite_sequences} is a standard ultrametric space. 
In such a case, 
we may omit the term ``generalized'', 
and speak simply of an \emph{ultrametric semilattice}.

Another example of a non-standard ultrametric semilattice of linear signals, 
one that is of particular interest to the study of timed computation, 
is the set of all discrete-event\footnote{
    A signal $s$ from $\pair{T}{\leq_T}$ to $V$ is \emph{discrete-event} 
    if and only if 
    there is an order-embedding of $\pair{\domain s}{\leq_{\domain s}}$ into $\pair{\naturals}{\leq_\naturals}$, 
    where $\leq_{\domain s}$ is the restriction of $\leq_T$ to $\domain s$.
} real-time signals over some non-empty set of values, 
equipped with the standard prefix relation
and the so-called ``Cantor metric'' 
(e.g., see \cite{Lee1998AFFCMOC}, \cite{Lee1999MCRPUDE}). 

\begin{example} \label{exam:discrete-event_real-time_signals}
    Let $V$ be a non-empty set. 
    
    Let $\structureA$ be a $\Sigma$-structure 
    such that $\carrier{\structureA}_\sortA$ is the set of all discrete-event signals from $\pair{\reals}{\leq_\reals}$ to $V$,\footnote{
        We write $\reals$ for the set of all real numbers, 
        and $\leq_\reals$ for the standard order on $\reals$.
    } 
    $\carrier{\structureA}_\sortD = \reals_{\geq 0}$, 
    and the following are true:
    \begin{enumerate}
        
        \item
            ${\meetsymbol^\structureA}$ is a binary operation on $\carrier{\structureA}_\sortA$ 
            such that for every $s_1, s_2 \in \carrier{\structureA}_\sortA$, 
            $s_1 \meetsymbol^\structureA s_2$ is the greatest common prefix of $s_1$ and $s_2$;
            
        \item
            ${\leqsymbol^\structureA}$ is the standard order on $\reals_{\geq 0}$; 
        
        \item
            $\zerosymbol^\structureA = 0$;
            
        \item
            $\dsymbol^\structureA$ is a function from $\carrier{\structureA}_\sortA \times \carrier{\structureA}_\sortA$ to $\carrier{\structureA}_\sortD$ 
            such that for every $s_1, s_2 \in \carrier{\structureA}_\sortA$, 
            \[
                \dsymbol^\structureA(s_1, s_2) = 
                \begin{cases}
                    0                                                                   &\text{if $s_1 = s_2$;} \\
                    2^{- \min \Set{r}{$r \in \reals$ and $s_1(r) \not\simeq s_2(r)$}}   &\text{otherwise.}
                \end{cases}
            \]
        
    \end{enumerate}

    Notice that since the domain of every signal in $\carrier{\structureA}_\sortA$ is well ordered by $\leq_\reals$, 
    for every $s_1, s_2 \in \carrier{\structureA}_\sortA$, 
    $\Set{r}{$r \in \reals$ and $s_1(r) \not\simeq s_2(r)$}$ is also well ordered by $\leq_\reals$, 
    and thus, 
    $\min \Set{r}{$r \in \reals$ and $s_1(r) \not\simeq s_2(r)$}$ is well defined.
    
%     It is easy to verify that, 
%     since the domain of every signal in $\carrier{\structureA}_\sortA$ is well ordered by $\leq_\reals$,  
%     for every $s_1, s_2, s_3, s_4 \in \carrier{\structureA}_\sortA$, 
%     \[
%         \dsymbol^\structureA(s_1, s_2) \leqsymbol^\structureA \dsymbol^\structureA(s_3, s_4) 
%     \] 
%     if and only if 
%     \[
%         \mathrm{d}_{\signalset[\pair{\reals}{\leq_\reals}, V]}(s_1, s_2) \supseteq \mathrm{d}_{\signalset[\pair{\reals}{\leq_\reals}, V]}(s_3, s_4) \displaypunctuation{,}
%     \] 
%     and thus, 
%     $\structureA$ is an ultrametric semilattice. 
%     In fact, 
%     there is an embedding of $\structureA$ into $\structureS[\pair{\reals}{\leq_\reals}, V]$, 
%     as the reader may wish to verify.
    It is easy to verify that $\structureA$ is a directed-complete and spherically complete ultrametric semilattice.
\end{example}

Finally, 
we include an example from the field of logic programming. 
We assume familiarity with the basic concepts of logic programming 
(e.g., see \cite{Lloyd1987FOLP}). 
Our notation is based on \cite{Hitzler2003GMAUDLP}. 

\begin{example} \label{exam:Herbrand_interpretations}
    Let $P$ be a normal logic program. 
    
    Let $\alpha$ be a non-empty countable ordinal, 
    and $l$ a function from $\mathrm{H}_P$, 
    the Herbrand base of $P$, 
    to $\alpha$. 
    
    Let $\structureA$ be a $\Sigma$-structure 
    such that $\carrier{\structureA}_\sortA$ is the set of all subsets of $\mathrm{H}_P$, 
    $\carrier{\structureA}_\sortD = \alpha \union \set{\alpha}$, 
    and the following are true:
    \begin{enumerate}
        
        \item
            ${\meetsymbol^\structureA}$ is a binary operation on $\carrier{\structureA}_\sortA$ 
            such that for every $I_1, I_2 \in \carrier{\structureA}_\sortA$, 
            \[
                I_1 \meetsymbol^\structureA I_2 = \{A \mathrel{\vert} \pbox[t]{\textwidth}{
                    $A \in I_1$, 
                    $A \in I_2$, 
                    and for every $A'$ 
                    such that $l(A') \in l(A)$ or $l(A') = l(A)$, 
                    $A' \in I_1$ \\if and only if $A' \in I_2$\};
                }
            \]
            
        \item
            ${\leqsymbol^\structureA}$ is a binary relation on $\carrier{\structureA}_\sortD$ 
            such that for every $\beta, \gamma \in \carrier{\structureA}_\sortD$, 
            \[
                \beta \leqsymbol^\structureA \gamma \Iff \text{$\gamma \in \beta$ or $\beta = \gamma$} \displaypunctuation{.}
            \] 
        
        \item
            $\zerosymbol^\structureA = \alpha$;
            
        \item
            $\dsymbol^\structureA$ is a function from $\carrier{\structureA}_\sortA \times \carrier{\structureA}_\sortA$ to $\carrier{\structureA}_\sortD$ 
            such that for every $I_1, I_2 \in \carrier{\structureA}_\sortA$,  
            \[
                \dsymbol^\structureA(I_1, I_2) = \{\beta \mathrel{\vert} \pbox[t]{\textwidth}{
                    $\beta \in \alpha$, 
                    and for every $A$ 
                    such that $l(A') \in \beta$ or $l(A') = \beta$, 
                    $A' \in I_1$ if and only if \\$A' \in I_2$\}.
                }
            \]
        
    \end{enumerate}
    
    Let $\leq_a$ be a binary relation on $\alpha$ 
    such that for every $\beta, \gamma \in \alpha$, 
    \[
        \beta \leq_a \gamma \Iff \text{$\beta \in \gamma$ or $\beta = \gamma$} \displaypunctuation{.} 
    \]
    
    Clearly, 
    $\pair{\alpha}{\leq_\alpha}$ is an ordered set. 
    
%     It is easy to verify that $\structureA$ is isomorphic to a substructure of $\structureS[\pair{\alpha}{\leq_\alpha}, \powerset \mathrm{H}_P]$, 
%     and thus, 
%     $\structureA$ is a generalized ultrametric semilattice.
    It is easy to verify that $\structureA$ is a directed-complete and spherically complete generalized ultrametric semilattice.
\end{example}

%%%%%%%%%%%%%%%%%%%%%%%%%%%%%%%%%%%%%%%%%%%%%%%%%%%%%%%%%%%%%%%%%%%%%%%%%%%%%%%%

\section{Contracting and Strictly Contracting Functions} \label{sec:constracting_and_strictly_contracting_functions}

Assume a function $F$ on $\structureA$. 

We say that $F$ is \emph{contracting} 
if and only if 
for every $a_1, a_2 \in \carrier{\structureA}_\sortA$, 
\[
    \distance{\structureA}{F(a_1)}{F(a_2)} \leqsymbol^\structureA \distance{\structureA}{a_1}{a_2}\displaypunctuation{.}
\] 

In other words, 
a function is contracting just as long as the generalized distance between any two elements in the range of the function is smaller than or equal to that between the elements in the domain of the function that map to them.
Notice that, 
because $\leqsymbol^\structureA$ is not necessarily a total order, 
this is different, 
in general, 
from the generalized distance between any two elements in the domain of the function being no bigger than that between the elements in the range of the function that those map to, 
which is why we have opted for the term ``contracting'' over the term ``non-expanding''.

We say that $F$ is \emph{strictly contracting} 
if and only if 
for every $a_1, a_2 \in \carrier{\structureA}_\sortA$ 
such that $a_1 \neq a_2$, 
\[
    \distance{\structureA}{F(a_1)}{F(a_2)} \lsymbol^\structureA \distance{\structureA}{a_1}{a_2} \displaypunctuation{.}
\] 

The following is immediate:
\begin{proposition} \label{prop:if_strictly_contracting_then_contracting}
    If\/ $F$ is strictly contracting, 
    then\/ $F$ is contracting.
\end{proposition}

To return to Example~\ref{exam:discrete-event_real-time_signals},
the contracting and strictly contracting functions on the generalized ultrametric semilattice of all discrete-event real-time signals over $V$ are exactly the causal and strictly causal functions respectively on such signals
(see \cite{Matsikoudis2013TFTOSCF}, \cite{Matsikoudis2013aOFPOSCF}).
And in the case of Example~\ref{exam:Herbrand_interpretations},
if the normal logic program $P$ is a so-called ``locally hierarchical'' program,
then the level mapping $l$ can be chosen so that $P$ can be modelled as a strictly contracting function on $\structureA$
(see \cite{Hitzler2003GMAUDLP}).

Now, 
contracting functions need not have fixed points 
(e.g., see \cite[exam.\,3.4]{Matsikoudis2013TFTOSCF}). 
But what about strictly contracting functions? 

\begin{proposition} \label{prop:if_strictly_contracting_then_at_most_one_fixed_point}
    If\/ $F$ is strictly contracting, 
    then\/ $F$ has at most one fixed point.
\end{proposition}
\begin{proof}
    Suppose that $F$ is strictly contracting.
    
    Suppose, 
    toward contradiction, 
    that $a_1$ and $a_2$ are two distinct fixed points of $F$.
    Then 
    \[
        \distance{\structureA}{F(a_1)}{F(a_2)} = \distance{\structureA}{a_1}{a_2} \displaypunctuation{,}
    \]
    obtaining a contradiction.
    
    Thus, 
    $F$ has at most one fixed point.
\end{proof}

\begin{theorem} \label{thm:if_strictly_contracting_then_unique_fixed_point}
    If\/ $\structureA$ is spherically complete,
    then every strictly contracting function on\/ $\structureA$ has exactly one fixed point.
\end{theorem}

Theorem~\ref{thm:if_strictly_contracting_then_unique_fixed_point} follows immediately from the fixed-point theorem of Priess-Crampe and Ribenboim for strictly contracting functions on spherically complete generalized ultrametric spaces (see \cite[thm.\,1]{Priess-Crampe1993FPCAGPS}), 
which is sometimes, 
and perhaps a little too liberally, 
referred to as a generalization of the \emph{Banach Fixed-Point Theorem}. 
The following, 
which follows immediately from another theorem of Priess-Crampe and Ribenboim (e.g., see \emph{Banach's Fixed Point Theorem} in \cite{Schorner2003UFPTAA}), 
justifies the use of the stronger property of spherical completeness in place of the standard property of Cauchy-completeness used in the latter:
\begin{theorem} \label{thm:non-empty_spherically_complete_if_and_only_if_strictly_contracting_principle}
    If\/ $\pair{\carrier{\structureA}_\sortD}{\leqsymbol^\structureA}$ is totally ordered, 
    then\/ $\structureA$ is spherically complete 
    if and only if 
    every strictly contracting function on\/ $\structureA$ has a fixed point.
\end{theorem}

Note that the hypothesis of $\pair{\tagset}{\preceq}$ being totally ordered in Theorem~\ref{thm:non-empty_spherically_complete_if_and_only_if_strictly_contracting_principle} cannot be discarded 
(see \cite[thm.\,5.5 and exam.\,5.8]{Matsikoudis2013TFTOSCF}).  

%%%%%%%%%%%%%%%%%%%%%%%%%%%%%%%%%%%%%%%%%%%%%%%%%%%%%%%%%%%%%%%%%%%%%%%%%%%%%%%%

\section{Fixed-Point Theory} \label{sec:fixed-point_theory}

We now develop the rudiments of a constructive fixed-point theory for strictly contracting functions.

%%%%%%%%%%%%%%%%%%%%%%%%%%%%%%%%%%%%%%%%%%%%%%%%%%%%%%%%%%%%%%%%%%%%%%%%%%%%%%%%

\subsection{Existence} \label{subsec:existence}

We start by proving another fixed-point existence result for strictly contracting functions, 
which is similar to Theorem~\ref{thm:if_strictly_contracting_then_unique_fixed_point}, 
but has a different premise.
The proof is more like Naundorf's proof in \cite{Naundorf2000SCFHAUFP}, 
but, 
as also possible in the case of the existence part of Theorem~\ref{thm:if_strictly_contracting_then_unique_fixed_point}
(see \cite[p.\,229]{Priess-Crampe1993FPCAGPS}), 
our main theorem applies to a more general type of function.

Assume a function $F$ on $\structureA$.

We say that $F$ is \emph{strictly contracting on orbits} 
if and only if 
for every $a \in \carrier{\structureA}_\sortA$ 
such that $a \neq F(a)$, 
\[
    \distance{\structureA}{F(a)}{F(F(a))} \lsymbol^\structureA \distance{\structureA}{a}{F(a)} \displaypunctuation{.}
\] 

In other words, 
$F$ is strictly contracting on orbits just as long as the generalized distance between every two successive elements in the orbit\footnote{
    For every set $A$, 
    every function $f$ on $A$, 
    and any $a \in A$,
    the \emph{orbit} of $a$ under $f$ is the sequence $\Sequence{f^n(a)}{$n \in \omega$}$.
} of every $a \in \carrier{\structureA}_\sortA$ under $F$ gets smaller and smaller along the orbit.

The following is immediate:
\begin{proposition} \label{prop:if_strictly_contracting_then_stritly_contracting_on_orbits}
    If\/ $F$ is strictly contracting, 
    then\/ $F$ is strictly contracting on orbits.
\end{proposition}

\begin{theorem} \label{thm:strictly_contracting_on_orbits_fixed-point_existence}
    If\/ $\structureA$ is directed-complete, 
    then every contracting function on\/ $\structureA$ that is strictly contracting on orbits has a fixed point.
\end{theorem}

Before we embark on the proof of the theorem,
we prove two important lemmas that will be useful throughout this section.

For every function $F$ on $\structureA$, 
and every $a \in \carrier{\structureA}_\sortA$, 
we say that $a$ is a \emph{post-fixed point} of $F$ 
if and only if 
$a \sqsubseteqsymbol^\structureA F(a)$.

\begin{lemma} \label{lem:progression_of_post-fixed_points}
    For every contracting function\/ $F$ on\/ $\structureA$, 
    and every\/ $a \in \carrier{\structureA}_\sortA$, 
    the following are true:
    \begin{enumerate}

        \item \label{enum:progression_of_post-fixed_points_1}
            $F(a) \meetsymbol^\structureA F(F(a))$ is a post-fixed point of\/ $F$; 

        \item \label{enum:progression_of_post-fixed_points_2}
            if\/ $a$ is a post-fixed point of\/ $F$,
            then\/ $a \sqsubseteqsymbol^\structureA F(a) \meetsymbol^\structureA F(F(a))$.

    \end{enumerate}
\end{lemma}
\begin{proof}
    Assume a contracting function $F$ on $\structureA$,
    and $a \in \carrier{\structureA}_\sortA$.

    Since $F$ is contracting,
    by Definition~\ref{def:generalized_ultrametric_semilattice}.\ref{enum:generalized_ultrametric_semilattice_4_2},
    \begin{align*}
        \distance{\structureA}{F(F(a) \meetsymbol^\structureA F(F(a)))}{F(F(a))}    &\leqsymbol^\structureA \distance{\structureA}{F(a) \meetsymbol^\structureA F(F(a))}{F(a)} \\
                                                                                    &=                      \distance{\structureA}{F(a) \meetsymbol^\structureA F(F(a))}{F(a) \meetsymbol^\structureA F(a)} \\
                                                                                    &\leqsymbol^\structureA \distance{\structureA}{F(a)}{F(F(a))} \displaypunctuation{,}
    \end{align*}
    and thus,
    by Definition~\ref{def:generalized_ultrametric_semilattice}.\ref{enum:generalized_ultrametric_semilattice_4_1},
    \begin{align*}
        F(a) \meetsymbol^\structureA F(F(a))    &\sqsubseteqsymbol^\structureA  F(F(a) \meetsymbol^\structureA F(F(a))) \meetsymbol^\structureA F(F(a)) \\
                                                &\sqsubseteqsymbol^\structureA  F(F(a) \meetsymbol^\structureA F(F(a))) \displaypunctuation{.}
    \end{align*}

    Thus,
    \ref{enum:progression_of_post-fixed_points_1} is true.

    Suppose that $a \sqsubseteqsymbol^\structureA F(a)$.

    Since $F$ is contracting,
    \[
        \distance{\structureA}{F(a)}{F(F(a))} \leqsymbol^\structureA \distance{\structureA}{a}{F(a)} \displaypunctuation{,}
    \]
    and thus,
    by Definition~\ref{def:generalized_ultrametric_semilattice}.\ref{enum:generalized_ultrametric_semilattice_4_1},
    \[
        a \meetsymbol^\structureA F(a) \sqsubseteqsymbol^\structureA F(a) \meetsymbol^\structureA F(F(a)) \displaypunctuation{.}
    \]
    And since $a \sqsubseteqsymbol^\structureA F(a)$,
    $a \meetsymbol^\structureA F(a) = a$,
    and thus,
    \[
        a \sqsubseteqsymbol^\structureA F(a) \meetsymbol^\structureA F(F(a)) \displaypunctuation{.}
    \]

    Thus,
    \ref{enum:progression_of_post-fixed_points_2} is true.
\end{proof}

\begin{lemma} \label{lem:supremum_of_post-fixed_points}
    For every contracting function\/ $F$ on\/ $\structureA$,
    and any set\/ $P$ of post-fixed points of\/ $F$,
    if\/ $P$ has a least upper bound in\/ $\pair{\carrier{\structureA}_\sortA}{\sqsubseteqsymbol^\structureA}$,
    then\/ $\LUBsymbol^\structureA P$ is a post-fixed point of\/ $F$.
\end{lemma}
\begin{proof}
    Assume a contracting function $F$ on $\structureA$,
    and a set $P$ of post-fixed points of $F$ that has a least upper bound in $\pair{\carrier{\structureA}_\sortA}{\sqsubseteqsymbol^\structureA}$.

    Assume $a \in P$.

    Since $F$ is contracting,
    \begin{equation} \label{eqn:supremum_of_post-fixed_points_1}
        \distance{\structureA}{F(a)}{F(\LUBsymbol^\structureA P)} \leqsymbol^\structureA \distance{\structureA}{a}{\LUBsymbol^\structureA P} \displaypunctuation{.}
    \end{equation}
    By Definition~\ref{def:generalized_ultrametric_semilattice}.\ref{enum:generalized_ultrametric_semilattice_4_2} and (\ref{eqn:supremum_of_post-fixed_points_1}),
    \begin{equation} \label{eqn:supremum_of_post-fixed_points_2}
        \distance{\structureA}{(\LUBsymbol^\structureA P) \meetsymbol^\structureA F(a)}{(\LUBsymbol^\structureA P) \meetsymbol^\structureA F(\LUBsymbol^\structureA P)} \leqsymbol^\structureA \distance{\structureA}{a}{\LUBsymbol^\structureA P} \displaypunctuation{.}
    \end{equation}
    Also,
    since $a$ is a post-fixed point of $F$,
    by Definition~\ref{def:generalized_ultrametric_semilattice}.\ref{enum:generalized_ultrametric_semilattice_4_2},
    \begin{align} \label{eqn:supremum_of_post-fixed_points_3}
        \distance{\structureA}{a}{(\LUBsymbol^\structureA P) \meetsymbol^\structureA F(a)}    &=          \distance{\structureA}{F(a) \meetsymbol^\structureA a}{F(a) \meetsymbol^\structureA \LUBsymbol^\structureA P} \nonumber \\
                                                &\leqsymbol^\structureA  \distance{\structureA}{a}{\LUBsymbol^\structureA P} \displaypunctuation{.}
    \end{align}
    By (\ref{eqn:supremum_of_post-fixed_points_2}),
    (\ref{eqn:supremum_of_post-fixed_points_3}),
    and the generalized ultrametric inequality,
    \[
        \distance{\structureA}{a}{(\LUBsymbol^\structureA P) \meetsymbol^\structureA F(\LUBsymbol^\structureA P)} \leqsymbol^\structureA \distance{\structureA}{a}{\LUBsymbol^\structureA P} \displaypunctuation{.}
    \]
    Then,
    by the generalized ultrametric inequality,
    \[
        \distance{\structureA}{\LUBsymbol^\structureA P}{(\LUBsymbol^\structureA P) \meetsymbol^\structureA F(\LUBsymbol^\structureA P)} \leqsymbol^\structureA \distance{\structureA}{a}{\LUBsymbol^\structureA P} \displaypunctuation{,}
    \]
    and thus,
    by Definition~\ref{def:generalized_ultrametric_semilattice}.\ref{enum:generalized_ultrametric_semilattice_4_1},
    \begin{align*}
        a \meetsymbol^\structureA \LUBsymbol^\structureA P  &\sqsubseteqsymbol^\structureA  (\LUBsymbol^\structureA P) \meetsymbol^\structureA (\LUBsymbol^\structureA P) \meetsymbol^\structureA F(\LUBsymbol^\structureA P) \\
                                &=          (\LUBsymbol^\structureA P) \meetsymbol^\structureA F(\LUBsymbol^\structureA P) \displaypunctuation{.}
    \end{align*}
    However,
    since $a \in P$,
    $a \sqsubseteqsymbol^\structureA \LUBsymbol^\structureA P$,
    and thus,
    $a \meetsymbol^\structureA \LUBsymbol^\structureA P = a$.
    Thus,
    \begin{align*}
        a   &\sqsubseteqsymbol^\structureA  (\LUBsymbol^\structureA P) \meetsymbol^\structureA F(\LUBsymbol^\structureA P) \\
            &\sqsubseteqsymbol^\structureA  F(\LUBsymbol^\structureA P) \displaypunctuation{.}
    \end{align*}

    Thus,
    by generalization,
    $F(\LUBsymbol^\structureA P)$ is an upper bound of $P$ in $\pair{\carrier{\structureA}_\sortA}{\sqsubseteqsymbol^\structureA}$.
    And since $\LUBsymbol^\structureA P$ is the least upper bound of $P$ in $\pair{\carrier{\structureA}_\sortA}{\sqsubseteqsymbol^\structureA}$,
    $\LUBsymbol^\structureA P \sqsubseteqsymbol^\structureA F(\LUBsymbol^\structureA P)$.
    Thus,
    $\LUBsymbol^\structureA P$ is a post-fixed point of $F$.
\end{proof}

\begin{proof}[Proof of Theorem~\ref{thm:strictly_contracting_on_orbits_fixed-point_existence}]
    Suppose that $\structureA$ is directed-complete.
    
    Assume a contracting function $F$ on $\structureA$ that is strictly contracting on orbits.

    Let $P = \Set{a}{$a$ is a post-fixed point of $F$}$.

    Let $a$ be a member of $\carrier{\structureA}_\sortA$.

    By Lemma~\ref{lem:progression_of_post-fixed_points}.\ref{enum:progression_of_post-fixed_points_1},
    \[
        F(a) \meetsymbol^\structureA F(F(a)) \sqsubseteqsymbol^\structureA F(F(a) \meetsymbol^\structureA F(F(a))) \displaypunctuation{,}
    \]
    and thus,
    $P \neq \emptyset$.
    Then,
    by Kuratowski's Lemma
    (see \cite[sec.\,10.2]{Davey2002ITLAO}),
    every chain in $\pair{P}{\sqsubseteqsymbol^\structureA}$ is contained in a $\subset$-maximal chain in $\pair{P}{\sqsubseteqsymbol^\structureA}$.

    Let $C$ be a $\subset$-maximal chain in $\pair{P}{\sqsubseteqsymbol^\structureA}$.

    Since $\structureA$ is directed-complete,
    $C$ has a least upper bound in $\pair{\carrier{\structureA}_\sortA}{\sqsubseteqsymbol^\structureA}$.

    We claim that $\LUBsymbol^\structureA C$ is a fixed point of $F$.

    Suppose,
    toward contradiction,
    that $\LUBsymbol^\structureA C$ is not a fixed point of $F$.

    Let $x = F(\LUBsymbol^\structureA C) \meetsymbol^\structureA F(F(\LUBsymbol^\structureA C))$.

    By Lemma~\ref{lem:supremum_of_post-fixed_points},
    $\LUBsymbol^\structureA C \sqsubseteqsymbol^\structureA F(\LUBsymbol^\structureA C)$,
    and thus,
    by Lemma~\ref{lem:progression_of_post-fixed_points}.\ref{enum:progression_of_post-fixed_points_2},
    $\LUBsymbol^\structureA C \sqsubseteqsymbol^\structureA x$.

    Suppose,
    toward contradiction,
    that $\LUBsymbol^\structureA C = x$.
    Since $F$ is strictly contracting on orbits,
    and $\LUBsymbol^\structureA C$ is not a fixed point of $F$,
    \begin{equation} \label{eqn:strictly_contracting_on_orbits_fixed_point_existence_1}
        \distance{\structureA}{F(\LUBsymbol^\structureA C)}{F(F(\LUBsymbol^\structureA C))} \lsymbol^\structureA \distance{\structureA}{\LUBsymbol^\structureA C}{F(\LUBsymbol^\structureA C)} \displaypunctuation{.}
    \end{equation}
    However,
    since $x = F(\LUBsymbol^\structureA C) \meetsymbol^\structureA F(F(\LUBsymbol^\structureA C))$
    and $\LUBsymbol^\structureA C = x$,
    by Definition~\ref{def:generalized_ultrametric_semilattice}.\ref{enum:generalized_ultrametric_semilattice_4_2},
    \begin{align*}
        \distance{\structureA}{\LUBsymbol^\structureA C}{F(\LUBsymbol^\structureA C)}                          &=          \distance{\structureA}{F(\LUBsymbol^\structureA C)}{\LUBsymbol^\structureA C} \\
                                                                    &=          \distance{\structureA}{F(\LUBsymbol^\structureA C)}{F(\LUBsymbol^\structureA C) \meetsymbol^\structureA F(F(\LUBsymbol^\structureA C))} \\
                                                                    &=          \distance{\structureA}{F(\LUBsymbol^\structureA C) \meetsymbol^\structureA F(\LUBsymbol^\structureA C)}{F(\LUBsymbol^\structureA C) \meetsymbol^\structureA F(F(\LUBsymbol^\structureA C))} \\
                                                                    &\leqsymbol^\structureA  \distance{\structureA}{F(\LUBsymbol^\structureA C)}{F(F(\LUBsymbol^\structureA C))} \displaypunctuation{,}
    \end{align*}
    contrary to (\ref{eqn:strictly_contracting_on_orbits_fixed_point_existence_1}).

    Therefore,
    $\LUBsymbol^\structureA C \sqsubsetsymbol^\structureA x$.
    Thus,
    $x \not\in C$.
    And by Lemma~\ref{lem:progression_of_post-fixed_points}.\ref{enum:progression_of_post-fixed_points_1},
    $x \sqsubseteqsymbol^\structureA F(x)$,
    and thus,
    $x \in P$.
    Thus,
    $C \union \set{x}$ is a chain in $\pair{P}{\sqsubseteqsymbol^\structureA}$,
    and $C \subset C \union \set{x}$,
    contrary to $C$ being a $\subset$-maximal chain in $\pair{P}{\sqsubseteqsymbol^\structureA}$.

    Therefore,
    $\LUBsymbol^\structureA C$ is a fixed point of $F$.
\end{proof}

There are two things to notice here.
First,
the proof of Theorem~\ref{thm:strictly_contracting_on_orbits_fixed-point_existence} is inherently non-constructive,
overtly appealing to the Axiom of Choice through the use of Kuratowski's Lemma.
And second,
there need not be only one fixed point;
indeed,
the identity function on $\structureA$ is trivially causal and strictly contracting on orbits,
yet every element is a fixed point of it.

The following is immediate from Proposition~\ref{prop:if_strictly_contracting_then_contracting}, \ref{prop:if_strictly_contracting_then_at_most_one_fixed_point}, and \ref{prop:if_strictly_contracting_then_stritly_contracting_on_orbits}, 
and Theorem~\ref{thm:strictly_contracting_on_orbits_fixed-point_existence}:
\begin{theorem} \label{thm:strictly_contracting_fixed-point_existence}
    If\/ $\structureA$ is directed-complete, 
    then every strictly contracting function on\/ $\structureA$ has exactly one fixed point.
\end{theorem}

If $\structureA$ is directed-complete,
then for every strictly contracting function $F$ on $\structureA$, 
we write $\fix F$ for the unique fixed point of $F$.

The following is immediate from Theorem~\ref{thm:non-empty_spherically_complete_if_and_only_if_strictly_contracting_principle} and \ref{thm:strictly_contracting_fixed-point_existence}:
\begin{corollary} \label{cor:if_directed-complete_subsemilattice_then_spherically_complete}
    \sloppy 
    If\/ $\pair{\carrier{\structureA}_\sortD}{\leqsymbol^\structureA}$ is totally ordered, 
    then if\/ $\structureA$ is directed-complete, 
    then\/ $\structureA$ is spherically complete.
\end{corollary}

We note that the hypothesis of $\pair{\tagset}{\preceq}$ being totally ordered in Corollary~\ref{cor:if_directed-complete_subsemilattice_then_spherically_complete} cannot be discarded 
(see \cite[exam.\,5.8]{Matsikoudis2013TFTOSCF}). 
As a consequence, 
Theorem~\ref{thm:if_strictly_contracting_then_unique_fixed_point} and \ref{thm:strictly_contracting_fixed-point_existence} are incomparable with respect to deduction;
that is,
one cannot deduce Theorem~\ref{thm:strictly_contracting_fixed-point_existence} from Theorem~\ref{thm:if_strictly_contracting_then_unique_fixed_point},
nor Theorem~\ref{thm:if_strictly_contracting_then_unique_fixed_point} from Theorem~\ref{thm:strictly_contracting_fixed-point_existence}.

%%%%%%%%%%%%%%%%%%%%%%%%%%%%%%%%%%%%%%%%%%%%%%%%%%%%%%%%%%%%%%%%%%%%%%%%%%%%%%%%

\subsection{Construction} \label{subsec:construction}

Although theoretically pleasing,
mere existence of fixed points is practically moot.
Theorem~\ref{thm:strictly_contracting_on_orbits_fixed-point_existence} and \ref{thm:strictly_contracting_fixed-point_existence}, 
just like Theorem~\ref{thm:if_strictly_contracting_then_unique_fixed_point}, 
offer little if no means of deductive reasoning about the fixed points ascertained to exist.

But how are we to construct these fixed points? 
Theorem~A.2 and A.4 in \cite{Matsikoudis2013TFTOSCF} seem to render standard fixed-point theories of ordered sets and metric spaces more or less irrelevant. 
At the same time, 
it may well be that the relevant fixed-point theorem of Priess-Crampe and Ribenboim is independent of the theory of generalized ultrametric spaces in the classical Zermelo-Fraenkel set theory without choice, 
thus lacking a constructive proof altogether.\footnote{
    A purportedly constructive proof for the fixed-point theorem of Priess-Crampe and Ribenboim under the hypothesis of a totally ordered set of distances was presented in \cite[thm.\,1.3.9]{Hitzler2001GMATILPS}. 
    However, 
    the proof covertly appeals to the Axiom of Choice through a potentially transfinite sequence of choices.
}

The answer lies in the non-constructive proof of Theorem~\ref{thm:strictly_contracting_on_orbits_fixed-point_existence}.
Indeed,
the proof contains all the ingredients of a transfinite recursion facilitating the construction of a chain that may effectively substitute for the maximal one only asserted to exist therein by an appeal to Kuratowski's Lemma.
% In fact, 
% Lemma~\ref{lem:progression_of_post-fixed_points} and \ref{lem:supremum_of_post-fixed_points} contain nearly all the ingredients of a transfinite recursion facilitating the construction of a chain that will converge to the desired fixed point.
We may start with any arbitrary post-fixed point of the function $F$, 
and iterate through the function $\lambda a : \carrier{\structureA}_\sortA \mathrel{.} F(a) \meetsymbol^\structureA F(F(a))$ to form an ascending chain of such points. 
Every so often, 
we may take the supremum of all post-fixed points theretofore constructed, 
and resume the process therefrom, 
until no further progress can be made. 
Of course, 
the phrase ``every so often'' is to be interpreted rather liberally here, 
and certain groundwork is required before we can formalize its transfinite intent.

We henceforth assume some familiarity with transfinite set theory, 
and in particular, 
ordinal numbers. 
The unversed reader may refer to any introductory textbook on set theory for details 
(e.g., see \cite{Enderton1977EOST}).

We write $\OmTsymbol^\structureA F$ for a function on $\structureA$,  
such that for any $a \in \carrier{\structureA}_\sortA$, 
\[
    (\OmTsymbol^\structureA F)(a) = F(a) \meetsymbol^\structureA F(F(a)) \displaypunctuation{.}
\]

In other words, 
$\OmTsymbol^\structureA F$ is the function $\lambda a : \carrier{\structureA}_\sortA \mathrel{.} F(a) \meetsymbol^\structureA F(F(a))$. 

Assume a post-fixed point $a$ of $F$.

We let 
\[
    {(\OmTsymbol^\structureA F)}^0(a) = a \displaypunctuation{,}
\] 
for every ordinal $\alpha$, 
\[
    {(\OmTsymbol^\structureA F)}^{\alpha + 1}(a) = (\OmTsymbol^\structureA F)({(\OmTsymbol^\structureA F)}^\alpha(a)) \displaypunctuation{,}
\]
and for every limit ordinal $\lambda$, 
\[
    {(\OmTsymbol^\structureA F)}^\lambda(a) = \LUBsymbol^{\structureA} \Set{{(\OmTsymbol^\structureA F)}^\alpha(a)}{$\alpha \in \lambda$} \displaypunctuation{.}
\]

The following implies that for every ordinal $\alpha$, 
${(\OmTsymbol^\structureA F)}^\alpha(a)$ is well defined:
\begin{lemma} \label{lem:monotone_operation}
    If\/ $\structureA$ is directed-complete,
    then for every contracting function\/ $F$ on\/ $\structureA$,
    any post-fixed point\/ $a$ of\/ $F$,
    and every ordinal\/ $\alpha$,
    \begin{enumerate}

        \item \label{enum:monotone_operation_1}
            ${(\OmTsymbol^\structureA F)}^\alpha(a) \sqsubseteqsymbol^\structureA F({(\OmTsymbol^\structureA F)}^\alpha(a))$;

        \item \label{enum:monotone_operation_2}
            for any\/ $\beta \in \alpha$,
            ${(\OmTsymbol^\structureA F)}^\beta(a) \sqsubseteqsymbol^\structureA {(\OmTsymbol^\structureA F)}^\alpha(a)$.

    \end{enumerate}
\end{lemma}
\begin{proof}
    Suppose that $\structureA$ is directed-complete.

    Assume a contracting function $F$ on $\structureA$,
    a post-fixed point $a$ of $F$,
    and an ordinal $\alpha$.

    We use transfinite induction on the ordinal $\alpha$ to jointly prove that \ref{enum:monotone_operation_1} and \ref{enum:monotone_operation_2} are true.

    If $\alpha = 0$,
    then ${(\OmTsymbol^\structureA F)}^\alpha(a) = a$.
    Thus,
    \ref{enum:monotone_operation_1} is trivially true,
    whereas \ref{enum:monotone_operation_2} is vacuously true.

    Suppose that there is an ordinal $\beta$
    such that $\alpha = \beta + 1$.

    Then
    \begin{align} \label{eqn:monotone_operation_1}
        {(\OmTsymbol^\structureA F)}^\alpha(a)    &=  (\OmTsymbol^\structureA F)({(\OmTsymbol^\structureA F)}^\beta(a)) \nonumber \\
                                &=  F({(\OmTsymbol^\structureA F)}^\beta(a)) \meetsymbol^\structureA F(F({(\OmTsymbol^\structureA F)}^\beta(a))) \displaypunctuation{.}
    \end{align}

    Thus,
    by Lemma~\ref{lem:progression_of_post-fixed_points}.\ref{enum:progression_of_post-fixed_points_1},
    \ref{enum:monotone_operation_1} is true.

    For every $\gamma \in \alpha$,
    either $\gamma = \beta$,
    or $\gamma \in \beta$,
    and thus,
    by the induction hypothesis,
    \begin{equation} \label{eqn:monotone_operation_2}
        {(\OmTsymbol^\structureA F)}^\gamma(a) \sqsubseteqsymbol^\structureA {(\OmTsymbol^\structureA F)}^\beta(a) \displaypunctuation{.}
    \end{equation}
    Also,
    by the induction hypothesis,
    \[
        {(\OmTsymbol^\structureA F)}^\beta(a) \sqsubseteqsymbol^\structureA F({(\OmTsymbol^\structureA F)}^\beta(a)) \displaypunctuation{.}
    \]
    Thus,
    by Lemma~\ref{lem:progression_of_post-fixed_points}.\ref{enum:progression_of_post-fixed_points_2} and (\ref{eqn:monotone_operation_1}),
    \begin{align} \label{eqn:monotone_operation_3}
        {(\OmTsymbol^\structureA F)}^\beta(a) &\sqsubseteqsymbol^\structureA  F({(\OmTsymbol^\structureA F)}^\beta(a)) \meetsymbol^\structureA F(F({(\OmTsymbol^\structureA F)}^\beta(a))) \nonumber \\
                            &=          {(\OmTsymbol^\structureA F)}^\alpha(a) \displaypunctuation{.}
    \end{align}
    And by (\ref{eqn:monotone_operation_2})
    and (\ref{eqn:monotone_operation_3}),
    ${(\OmTsymbol^\structureA F)}^\gamma(a) \sqsubseteqsymbol^\structureA {(\OmTsymbol^\structureA F)}^\alpha(a)$.
    Thus,
    \ref{enum:monotone_operation_2} is true.

    Otherwise,
    $\alpha$ is a limit ordinal.
    By the induction hypothesis,
    $\pair{\Set{{(\OmTsymbol^\structureA F)}^\beta(a)}{$\beta \in \alpha$}}{\sqsubseteqsymbol^\structureA}$ is totally ordered,
    and thus,
    $\Set{{(\OmTsymbol^\structureA F)}^\beta(a)}{$\beta \in \alpha$}$ is directed in $\pair{\carrier{\structureA}_\sortA}{\sqsubseteqsymbol^\structureA}$.
    And since $\structureA$ is directed-complete,
    $\Set{{(\OmTsymbol^\structureA F)}^\beta(a)}{$\beta \in \alpha$}$ has a least upper bound in $\pair{\carrier{\structureA}_\sortA}{\sqsubseteqsymbol^\structureA}$,
    and
    \[
        {(\OmTsymbol^\structureA F)}^\alpha(a) = \LUBsymbol^\structureA\Set{{(\OmTsymbol^\structureA F)}^\beta(a)}{$\beta \in \alpha$} \displaypunctuation{.}
    \]

    Thus,
    \ref{enum:monotone_operation_2} is trivially true.

    By the induction hypothesis,
    for every $\beta \in \alpha$,
    ${(\OmTsymbol^\structureA F)}^\beta(a) \sqsubseteqsymbol^\structureA F({(\OmTsymbol^\structureA F)}^\beta(a))$.
    Thus,
    by Lemma~\ref{lem:supremum_of_post-fixed_points},
    \ref{enum:monotone_operation_1} is true.
\end{proof}

By Lemma~\ref{lem:monotone_operation}.\ref{enum:monotone_operation_2}, 
and a simple cardinality argument, 
there is an ordinal $\alpha$
such that for every ordinal $\beta$ 
such that $\alpha \in \beta$,
${(\OmTsymbol^\structureA F)}^\beta(a) = {(\OmTsymbol^\structureA F)}^\alpha(a)$.
In fact, 
there is a least ordinal $\alpha$
such that for every contracting function $F$ on $\structureA$, 
any post-fixed point $a$ of $F$, 
and every ordinal $\beta$ 
such that $\alpha \in \beta$,
${(\OmTsymbol^\structureA F)}^\beta(a) = {(\OmTsymbol^\structureA F)}^\alpha(a)$.

We write $\oh \structureA$ for the least ordinal $\alpha$ 
such that there is no function $\varphi$ from $\alpha$ to $\carrier{\structureA}_\sortA$ 
such that for every $\beta, \gamma \in \alpha$, 
if $\beta \in \gamma$, 
then $\varphi(\beta) \sqsubsetsymbol^\structureA \varphi(\gamma)$. 

In other words, 
$\oh \structureA$ is the least ordinal that cannot be orderly embedded in $\pair{\carrier{\structureA}_\sortA}{\sqsubseteqsymbol^\structureA}$, 
which we may think of as the \emph{ordinal height} of $\structureA$. 
Notice that the Hartogs number of $\carrier{\structureA}_\sortA$ is an ordinal that cannot be orderly embedded in $\pair{\carrier{\structureA}_\sortA}{\sqsubseteqsymbol^\structureA}$, 
and thus, 
$\oh \structureA$ is well defined, 
and in particular, 
smaller than or equal to the Hartogs number of $\carrier{\structureA}_\sortA$.

\begin{lemma} \label{lem:ordinal_height}
    If\/ $\structureA$ is directed-complete,
    then for every contracting function\/ $F$ on\/ $\structureA$,
    any post-fixed point\/ $a$ of\/ $F$, 
    and every ordinal\/ $\alpha$,
    if\/ ${(\OmTsymbol^\structureA F)}^\alpha(a)$ is not a fixed point of\/ $\OmTsymbol^\structureA F$, 
    then\/ $\alpha + 2 \in \oh \structureA$.
\end{lemma}
\begin{proof}
    Suppose that $\structureA$ is directed-complete.

    Assume a contracting function $F$ on $\structureA$,
    a post-fixed point $a$ of $F$,
    and an ordinal $\alpha$.

    Suppose that ${(\OmTsymbol^\structureA F)}^\alpha(a)$ is not a fixed point of $\OmTsymbol^\structureA F$.

    We claim that for any $\beta, \gamma \in \alpha + 2$,
    if $\beta \neq \gamma$,
    then
    \[
        {(\OmTsymbol^\structureA F)}^\beta(a) \neq {(\OmTsymbol^\structureA F)}^\gamma(a) \displaypunctuation{.}
    \]

    Suppose,
    toward contradiction,
    that there are $\beta, \gamma \in \alpha + 2$
    such that $\beta \neq \gamma$,
    but
    \[
        {(\OmTsymbol^\structureA F)}^\beta(a) = {(\OmTsymbol^\structureA F)}^\gamma(a) \displaypunctuation{.}
    \]
    Without loss of generality,
    assume that $\beta \in \gamma$.
    Since $F$ is contracting,
    by Lemma~\ref{lem:monotone_operation}.\ref{enum:monotone_operation_2},
    \begin{align*}
        {(\OmTsymbol^\structureA F)}^\beta(a)   &\sqsubseteqsymbol^\structureA  {(\OmTsymbol^\structureA F)}^{\beta + 1}(a) \\
                            &\sqsubseteqsymbol^\structureA  {(\OmTsymbol^\structureA F)}^\gamma(a) \displaypunctuation{,}
    \end{align*}
    and thus,
    \[
        {(\OmTsymbol^\structureA F)}^\beta(a) = {(\OmTsymbol^\structureA F)}^{\beta + 1}(a) \displaypunctuation{.}
    \]
    And since $\beta \in \gamma \in \alpha + 2$,
    either $\beta \in \alpha$,
    or $\beta = \alpha$.
    Thus,
    by an easy transfinite induction,
    \[
        {(\OmTsymbol^\structureA F)}^\beta(a) = {(\OmTsymbol^\structureA F)}^\alpha(a) \displaypunctuation{,}
    \]
    contrary to the assumption that ${(\OmTsymbol^\structureA F)}^\alpha(a)$ is not a fixed point of $\OmTsymbol^\structureA F$.

    Therefore,
    for any $\beta, \gamma \in \alpha + 2$,
    \[
        {(\OmTsymbol^\structureA F)}^\beta(a) = {(\OmTsymbol^\structureA F)}^\gamma(a)
    \]
    if and only if
    $\beta = \gamma$.
    Thus,
    since $F$ is contracting,
    by Lemma~\ref{lem:monotone_operation}.\ref{enum:monotone_operation_2},
    there is a function $\varphi$ from $\alpha + 2$ to $\carrier{\structureA}_\sortA$
    such that for every $\beta, \gamma \in \alpha + 2$,
    if $\beta \in \gamma$,
    then $\varphi(\beta) \sqsubset^\structureA \varphi(\gamma)$.
    Thus,
    by definition of $\oh \structureA$,
    $\alpha + 2 \in \oh \structureA$.
\end{proof}

By Lemma~\ref{lem:ordinal_height}, 
${(\OmTsymbol^\structureA F)}^{\oh \structureA}(a)$ is a fixed point of $\OmTsymbol^\structureA F$. 
Nevertheless, 
${(\OmTsymbol^\structureA F)}^{\oh \structureA}(a)$ need not be a fixed point of $F$ as intended.
Indeed, 
the recursion process might start stuttering at points that are not fixed under the function in question 
(e.g., see \cite[exam.3.4]{Matsikoudis2013TFTOSCF}).
If the function is strictly contracting on orbits, 
however, 
progress at such points is guaranteed.

\begin{lemma} \label{lem:fixed-point_preservation}
    For every function\/ $F$ on\/ $\structureA$ that is strictly contracting on orbits,
    $a$ is a fixed point of\/ $F$
    if and only if\/ 
    $a$ is a fixed point of\/ $\OmTsymbol^\structureA F$. 
\end{lemma}

\begin{proof}
    Assume a function $F$ on $\structureA$ that is strictly contracting on orbits.

    If $a$ is a fixed point of $F$,
    then
    \begin{align*}
        a   &=  F(a) \\
            &=  F(F(a)) \displaypunctuation{,}
    \end{align*}
    and thus,
    \begin{align*}
        a   &=  F(a) \meetsymbol^\structureA F(F(a)) \\
            &=  (\OmTsymbol^\structureA F)(a) \displaypunctuation{.}
    \end{align*}

    Conversely,
    suppose that $a$ is a fixed point of $\OmTsymbol^\structureA F$.

    Then,
    by Definition~\ref{def:generalized_ultrametric_semilattice}.\ref{enum:generalized_ultrametric_semilattice_4_2},
    \begin{align} \label{eqn:fixed-point_preservation_1}
        \distance{\structureA}{a}{F(a)} &=                      \distance{\structureA}{(\OmTsymbol^\structureA F)(a)}{F(a)} \nonumber \\
                                        &=                      \distance{\structureA}{F(a) \meetsymbol^\structureA F(F(a))}{F(a)} \nonumber \\
                                        &=                      \distance{\structureA}{F(a) \meetsymbol^\structureA F(F(a))}{F(a) \meetsymbol^\structureA F(a)} \nonumber \\
                                        &\leqsymbol^\structureA \distance{\structureA}{F(a)}{F(F(a))} \displaypunctuation{.}
    \end{align}

    Suppose,
    toward contradiction,
    that $a$ is not a fixed point of $F$.
    Then,
    since $F$ is strictly contracting on orbits,
    \[
        \distance{\structureA}{F(a)}{F(F(a))} \lsymbol^\structureA \distance{\structureA}{a}{F(a)} \displaypunctuation{,}
    \]
    contrary to (\ref{eqn:fixed-point_preservation_1}).

    Therefore,
    $a$ is a fixed point of $F$.
\end{proof}

We may at last put all the different pieces together to obtain a constructive version of Theorem~\ref{thm:strictly_contracting_on_orbits_fixed-point_existence}.

\begin{theorem} \label{thm:strictly_contracting_on_orbits_fixed_point}
    If\/ $\structureA$ is directed-complete,
    then for every contracting function\/ $F$ on\/ $\structureA$ that is strictly contracting on orbits,
    and any post-fixed point\/ $a$ of\/ $F$,
%     and every ordinal $\alpha$ not dominated\footnote{
% %         A set $A$ is dominated by a set $B$
% %         if and only if
% %         there is a one-to-one function from $A$ to $B$.
%     } by $\carrier{\structureA}_\sortA$,
    ${(\OmTsymbol^\structureA F)}^{\oh \structureA}(a)$ is a fixed point of\/ $F$.
\end{theorem}
\begin{proof}
    Suppose that $\structureA$ is directed-complete.

    Assume a contracting function $F$ on $\structureA$ that is strictly contracting on orbits,
    and a post-fixed point $a$ of $F$.

    Suppose,
    toward contradiction,
    that ${(\OmTsymbol^\structureA F)}^{\oh \structureA}(a)$ is not a fixed point of $\OmTsymbol^\structureA F$.
    Then,
    by Lemma~\ref{lem:ordinal_height},
    $\oh \structureA + 2 \in \oh \structureA$,
    a contradiction.

    Therefore,
    ${(\OmTsymbol^\structureA F)}^{\oh \structureA}(a)$ is a fixed point of $\OmTsymbol^\structureA F$.
    And since $F$ is strictly contracting on orbits,
    by Lemma~\ref{lem:fixed-point_preservation},
    ${(\OmTsymbol^\structureA F)}^{\oh \structureA}(a)$ is a fixed point of $F$.
\end{proof}

To be pedantic, 
Theorem~\ref{thm:strictly_contracting_on_orbits_fixed_point} does not directly prove that $F$ has a fixed point; 
unless there is a post-fixed point of $F$, 
the theorem is true vacuously. 
But by Lemma~\ref{lem:progression_of_post-fixed_points}.\ref{enum:progression_of_post-fixed_points_1}, 
for every $a \in \carrier{\structureA}_\sortA$, 
$(\OmTsymbol^\structureA F)(a)$ is a post-fixed point of $F$.

The following is immediate from Proposition~\ref{prop:if_strictly_contracting_then_contracting} and \ref{prop:if_strictly_contracting_then_stritly_contracting_on_orbits}, 
Lemma~\ref{lem:progression_of_post-fixed_points}.\ref{enum:progression_of_post-fixed_points_1}, 
and Theorem~\ref{thm:strictly_contracting_on_orbits_fixed_point}:
\begin{theorem} \label{thm:strictly_contracting_fixed_point_I}
    If\/ $\structureA$ is directed-complete,
    then for every strictly contracting function\/ $F$ on\/ $\structureA$, 
    and every\/ $a \in \carrier{\structureA}_\sortA$,
    \[
        \fix F = {(\OmTsymbol^\structureA F)}^{\oh \structureA}((\OmTsymbol^\structureA F)(a)) \displaypunctuation{.}
    \]
\end{theorem}

This construction of fixed points as ``limits of stationary transfinite iteration sequences'' is very similar to the construction of extremal fixed points of monotone operators in \cite{Cousot1979CVOTFPT} and references therein, 
where the function iterated is not $\OmTsymbol^\structureA F$, 
but $F$ itself.
Notice, 
however, 
that if $F$ preserves $\sqsubseteqsymbol^\structureA$, 
then for any post-fixed point $a$ of $F$,
$(\OmTsymbol^\structureA F)(a) = F(a)$.

The astute reader will at this point anticipate the following:
\begin{theorem} \label{thm:strictly_contracting_fixed_point_II}
    If\/ $\structureA$ is directed-complete,
    then for every strictly contracting function\/ $F$ on\/ $\structureA$,
    \[
        \fix F = \LUBsymbol^\structureA \Set{a}{$a$ is a post-fixed point of $F$} \displaypunctuation{.}
    \]
\end{theorem}

\begin{proof}
    Suppose that $\structureA$ is directed-complete.

    Assume a strictly contracting function $F$ on $\structureA$.

    Assume a post-fixed point $a$ of $F$.

    By Lemma~\ref{lem:monotone_operation}.\ref{enum:monotone_operation_2},
    $a \sqsubseteqsymbol^\structureA {(\OmTsymbol^\structureA F)}^{\oh \structureA}(a)$,
    and thus,
    since $F$ is strictly contracting,
    by Proposition~\ref{prop:if_strictly_contracting_then_contracting} and \ref{prop:if_strictly_contracting_then_stritly_contracting_on_orbits},
    Lemma~\ref{lem:monotone_operation}.\ref{enum:monotone_operation_2},
    and Theorem~\ref{thm:strictly_contracting_on_orbits_fixed_point},
    $a \sqsubseteqsymbol^\structureA \fix F$.

    Thus,
    by generalization,
    $\fix F$ is an upper bound of $\Set{a}{$a$ is a post-fixed point of $F$}$ in $\pair{\carrier{\structureA}_\sortA}{\sqsubseteqsymbol^\structureA}$.
    And since $\fix F$ is a post-fixed point of $F$,
    for every upper bound $u$ of $\Set{a}{$a$ is a post-fixed point of $F$}$ in $\pair{\carrier{\structureA}_\sortA}{\sqsubseteqsymbol^\structureA}$,
    $\fix F \sqsubseteqsymbol^\structureA u$.
    Thus,
    \[
        \fix F = \LUBsymbol^\structureA\Set{a}{$a$ is a post-fixed point of $F$} \displaypunctuation{.} \qedhere
    \]
\end{proof}

In retrospect,
we find that Theorem~\ref{thm:strictly_contracting_fixed_point_II}
may be derived directly from first principles.
In particular,
and under the premise of the corollary,
it is easy to establish without any use of Theorem~\ref{thm:strictly_contracting_on_orbits_fixed_point} that for every $a \in \carrier{\structureA}_\sortA$,
$a \sqsubseteqsymbol^\structureA \fix F$
if and only if
$a \sqsubseteqsymbol^\structureA F(a)$, 
as the reader may wish to verify.

The construction of Theorem~\ref{thm:strictly_contracting_fixed_point_II} is identical in form to Tarski's well known construction of greatest fixed points of order-preserving functions on complete lattices 
(see \cite[thm.\,1]{Tarski1955ALFTAIA}).

Finally, 
we note that $\OmTsymbol^\structureA F$ is not, 
in general, 
order-preserving under the above premises 
(see \cite[exam.\,2.15]{Matsikoudis2013TFTOSCF}), 
as might be suspected, 
and thus, 
our fixed-point theorem is not a reduction to a standard order-theoretic one.

In view of Example~\ref{exam:discrete-event_real-time_signals} and \ref{exam:Herbrand_interpretations}, 
and the comments in the paragraph following Proposition~\ref{prop:if_strictly_contracting_then_contracting}, 
Theorem~\ref{thm:strictly_contracting_fixed_point_I} and \ref{thm:strictly_contracting_fixed_point_II} can be directly applied to study the behaviour of strictly causal discrete-event components in feedback
(see \cite{Matsikoudis2013TFTOSCF}, \cite{Matsikoudis2013aOFPOSCF}), 
and obtain, 
constructively,
the unique supported model of locally hierarchical normal logic programs
(see \cite{Hitzler2003GMAUDLP}).

%%%%%%%%%%%%%%%%%%%%%%%%%%%%%%%%%%%%%%%%%%%%%%%%%%%%%%%%%%%%%%%%%%%%%%%%%%%%%%%%

\subsection{Induction} \label{subsec:induction}

Having used transfinite recursion to construct fixed points, 
we may use transfinite induction to prove properties of them.
And in the case of strictly contracting endofunctions, 
which have exactly one fixed point, 
we may use Theorem~\ref{thm:strictly_contracting_fixed_point_I} to establish a special proof rule.

Assume $P \subseteq \carrier{\structureA}_\sortA$.

We say that $P$ is \emph{strictly inductive} 
if and only if 
every non-empty chain in $\pair{P}{\sqsubseteqsymbol^\structureA}$ has a least upper bound in $\pair{P}{\sqsubseteqsymbol^\structureA}$.

Note that $P$ is strictly inductive 
if and only if 
$\pair{P}{\sqsubseteqsymbol^\structureA}$ is directed-complete
(see \cite[cor.\,2]{Markowsky1976CPADSWA}). 

\begin{theorem} \label{thm:fixed-point_induction}
    If\/ $\structureA$ is directed-complete,
    then for every strictly contracting function\/ $F$ on\/ $\structureA$,
    and every non-empty, strictly inductive\/ $P \subseteq \carrier{\structureA}_\sortA$, 
    if for every\/ $a \in P$, 
    $(\OmTsymbol^\structureA F)(a) \in P$, 
    then\/ $\fix F \in P$.
\end{theorem}
\begin{proof}
    Suppose that $\structureA$ is directed-complete.

    Assume a strictly contracting function $F$ on $\structureA$,
    and non-empty, strictly inductive $P \subseteq \carrier{\structureA}_\sortA$.

    Suppose that for every $a \in P$,
    $(\OmTsymbol^\structureA F)(a) \in P$.

    Let $a$ be a member of $P$.

    By Lemma~\ref{lem:progression_of_post-fixed_points}.\ref{enum:progression_of_post-fixed_points_1},
    $(\OmTsymbol^\structureA F)(a)$ is a post-fixed point of $F$.

    We use transfinite induction to prove that for every ordinal $\alpha$,
    ${(\OmTsymbol^\structureA F)}^\alpha((\OmTsymbol^\structureA F)(a)) \in P$.

    If $\alpha = 0$,
    then
    \[
        {(\OmTsymbol^\structureA F)}^\alpha((\OmTsymbol^\structureA F)(a)) = (\OmTsymbol^\structureA F)(a) \displaypunctuation{,}
    \]
    and thus,
    since $P$ is closed under $\OmTsymbol^\structureA F$,
    ${(\OmTsymbol^\structureA F)}^\alpha((\OmTsymbol^\structureA F)(a)) \in P$.

    If there is an ordinal $\beta$
    such that $\alpha = \beta + 1$,
    then
    \[
        {(\OmTsymbol^\structureA F)}^\alpha((\OmTsymbol^\structureA F)(a)) = (\OmTsymbol^\structureA F)({(\OmTsymbol^\structureA F)}^\beta((\OmTsymbol^\structureA F)(a))) \displaypunctuation{.}
    \]
    By the induction hypothesis,
    ${(\OmTsymbol^\structureA F)}^\beta((\OmTsymbol^\structureA F)(a)) \in P$,
    and thus,
    since $P$ is closed under $\OmTsymbol^\structureA F$,
    ${(\OmTsymbol^\structureA F)}^\alpha((\OmTsymbol^\structureA F)(a)) \in P$.

    Otherwise,
    $\alpha$ is a limit ordinal,
    and thus,
    \[
        {(\OmTsymbol^\structureA F)}^\alpha((\OmTsymbol^\structureA F)(a)) = \LUBsymbol^\structureA \Set{{(\OmTsymbol^\structureA F)}^\beta((\OmTsymbol^\structureA F)(a))}{$\beta \in \alpha$} \displaypunctuation{.}
    \]
    By the induction hypothesis,
    \[
        \Set{{(\OmTsymbol^\structureA F)}^\beta((\OmTsymbol^\structureA F)(a))}{$\beta \in \alpha$} \subseteq P \displaypunctuation{,}
    \]
    and by Lemma~\ref{lem:monotone_operation}.\ref{enum:monotone_operation_2},
    $\pair{\Set{{(\OmTsymbol^\structureA F)}^\beta((\OmTsymbol^\structureA F)(a))}{$\beta \in \alpha$}}{\sqsubseteqsymbol^\structureA}$ is totally ordered.
    Thus,
    since $P$ is strictly inductive,
    ${(\OmTsymbol^\structureA F)}^\alpha((\OmTsymbol^\structureA F)(a)) \in P$.

    Therefore,
    by transfinite induction,
    for every ordinal $\alpha$,
    ${(\OmTsymbol^\structureA F)}^\alpha((\OmTsymbol^\structureA F)(a)) \in P$.

    By Theorem~\ref{thm:strictly_contracting_fixed_point_I},
    \[
        \fix F = {(\OmTsymbol^\structureA F)}^{\oh \structureA}((\OmTsymbol^\structureA F)(a)) \displaypunctuation{,}
    \]
    and thus,
    $\fix F \in P$.
\end{proof}

Theorem~\ref{thm:fixed-point_induction} is an induction principle that one may use to prove properties of fixed points of strictly contracting endofunctions. 
We think of properties extensionally here; 
that is, 
a property is a subset of $\carrier{\structureA}_\sortA$. 
And the properties that are admissible for use with this principle are those that are non-empty and strictly inductive. 
According to the principle, 
then, 
for every strictly contracting function $F$ on any directed-complete generalized ultrametric semilattice $\structureA$,
every non-empty, strictly inductive property that is preserved by $\OmTsymbol^\structureA F$ is true of $\fix F$.

We refer to \cite[sec.\,5.3]{Matsikoudis2013TFTOSCF} for a comparison between this principle with the fixed-point induction principle for order-preserving functions on complete partial orders
(see \cite{Scott1969ATOP}),
and the fixed-point induction principle for contraction mappings on complete metric spaces
(see \cite{Reed1986ATMFCSP}, \cite{Rounds1985AOTTSOCP}, \cite{Roscoe1991TCSATMOC}, \cite{Kozen2007AOMC}).

%%%%%%%%%%%%%%%%%%%%%%%%%%%%%%%%%%%%%%%%%%%%%%%%%%%%%%%%%%%%%%%%%%%%%%%%%%%%%%%%

% \nocite{*}

% \bibliography{/Users/ematsi/Documents/Research/Bibliography/Bibliography,/Users/ematsi/Documents/Research/Library/Library}
\bibliography{document}

\bibliographystyle{eptcs}

%%%%%%%%%%%%%%%%%%%%%%%%%%%%%%%%%%%%%%%%%%%%%%%%%%%%%%%%%%%%%%%%%%%%%%%%%%%%%%%%

\end{document}

%%%%%%%%%%%%%%%%%%%%%%%%%%%%%%%%%%%%%%%%%%%%%%%%%%%%%%%%%%%%%%%%%%%%%%%%%%%%%%%%